\documentclass[12pt]{amsart}
\usepackage{amsmath,amsthm,amscd,amsfonts,amssymb}
\usepackage[all]{xy}
\usepackage[ansinew]{inputenc} 
\usepackage{datetime}

\newtheorem{thm}{Theorem}[section]

\newtheorem{prop}[thm]{Proposition}
\theoremstyle{definition}
\newtheorem{defi}[thm]{Definition}
\newtheorem{obs}[thm]{Remark}

\newtheorem{lemma}[thm]{Lemma}

\newtheorem*{acknowledgement*}{Acknowledgement}

\numberwithin{equation}{section}

\newcommand{\R}{{\mathbb R}}
\newcommand{\la}{\langle}
\newcommand{\ra}{\rangle}
\newcommand{\cC}{{\mathcal C}}
\renewcommand{\a}{\alpha}
\def\bemol{{\boldsymbol{\flat}}}
\def\C{\mathcal {C}}

\def\iiota{{\text{\larger[1]{\emph\i}}}}

\begin{document}

\title[Classical and relativistic fluids as intermediate integrals]
 {Classical and relativistic fluids as intermediate integrals of finite dimensional mechanical systems}

\author{R. J.  Alonso-Blanco}

\address{Departamento de Matem\'{a}ticas, Universidad de Salamanca, Plaza de la Merced 1-4, E-37008 Salamanca,  Spain.}
\email{ricardo@usal.es}

\begin{abstract}
We explore the relationship between mechanical systems describing the motion of a particle with the mechanical systems describing a continuous medium. More specifically, we will study how the so-called intermediate integrals or fields of solutions of a finite dimensional mechanical system (a second order differential  equation) are simultaneously Euler's equations of fluids and conversely. This will be done both in the classical and relativistic context. A direct relationship will be established by means of  the so-called time constraint (classical unsteady case, static or not) and the relativistic correction (for arbitrary pseudo-Riemannian metrics).
\end{abstract}

\maketitle


\tableofcontents

\bigskip

\section{Introduction}

The initial idea of this work is the following simple observation: Euler's equation for a stationary fluid, say $(v\cdot\nabla)\,v=\nabla\,P/\rho$, is formally the same as the Newton's equation of a single particle subjected to a force $\nabla\,P/\rho$. The fundamental difference between the two is interpretation: In the first case, $v$ is a vector field defined through space and in the second one it is the velocity vector along a particle path. The first one is the equation for a certain continuum medium, while the second one is the equation for a single material point particle.

One way of understanding the relationship just mentioned is as follows: when a mechanical system defines the evolution of a material point particle, all possible \emph{virtual} paths that the particle could make are simultaneously being taken into account
; it turns out that the consideration of (compatible) beams of these virtual particles describes the evolution of true material fluids. It should be stressed that it is not about the consideration of systems made up of a large number of particles, followed by some process of taking limits but the \emph{physical realization of virtual motions of a single particle!} As we will see, this will be even more significant when considering the relativistic case.
\smallskip

What we will do is explore formal relationship, trying to see how far it can go in more general contexts.
\smallskip

The mathematical concept that gathers the previous scheme is that of \emph{intermediate integral} of a differential equation, and is well known and classical (see, for example \cite{Goursat}, where Goursat develops Monge's method for integration of certain second order partial differential equations). We now particularize it as follows: an intermediate integral of a mechanical system (a second order differential equation) on $M$, is, by definition, a tangent vector field (a first order differential equation) on  $M$ whose integral curves (solutions) are also solutions of that mechanical system. As a particular but very important instance, the Hamilton-Jacobi equation is the equation of the intermediate integrals which give at the same time Lagrangian submanifolds of $T^*M$ (see \cite{RM}).
An akin concept is the classical concept of  \emph{field of extremals} in the calculus of variations (see, for instance, \cite{Gelfand}, p. 131 ff). Another well known example is the following: given a metric or, more generally, a connection, let us consider the equation of its geodesic trajectories, which is a second order differential equation. Then, each tangent vector field whose integral curves are geodesic (a \emph{``geodesic vector field''}), constitutes an intermediate integral of the equation of the geodesics.
\medskip

 Since the concept of intermediate integral  produces  a unification in various geometric formalisms of Mechanics  we believe it deserves further consideration. Just as an example of other possible scenarios where such a concept could shed some light, we bring a passage from De Broglie where,  with a view to introducing Wave Mechanics, it seems clear that he is dealing with an intermediate integral:
``In order to examine this in detail let us consider not a single particle but a cloud of identical particles, all situated in the same field of force and without mutual reactions. The motion of this cloud, taken altogether, represents a whole assembly of possible motions of the same particle in the given field''(\cite{Broglie}, p. 27).
\medskip

The main objective in this work is to demonstrate that the basic equations of fluid mechanics (exception made of conservation laws that must be added independently) are obtainable as equations of intermediate integrals of ordinary finite mechanical systems.
\medskip

The results of our research are as follows:\medskip

Classical (non-relativistic) case: there is an identity between the equations of the intermediate integrals of the mechanical system produced by an \emph{integrable 1-form of force} and the Euler equations for steady fluids (in a configuration space equipped with an arbitrary pseudo-Riemannian metric). Additional forces may be included at convenience, depending on the model to be described. To obtain Euler equations in the unsteady case it is enough to start from the same type of mechanical systems (integrable forces), modify them by means of a time constraint (see below) and then take the equations of their intermediate integrals. In the subcase where the spatial metric (see below) is static, we arrive to Euler (and Bernouilli) equations for unsteady fluids including the classical versions in Euclidean space. When we consider a space-time metric of type Friedman-Lemaître-Robertson-Walker (without even needing to comply with Einstein's equations), we obtain non-relativistic Euler (and Bernoulli) equations within these background spaces, in this way recovering some of the equations found in the literature for specific FLRW models (see below).
\smallskip

Relativistic case: Euler relativistic equations for a perfect fluid are canonically obtained from a mechanical system defined by an integrable form of force, modifying said system through relativistic correction (see below) and then considering the equations of the intermediate integrals of the new mechanical system (now already being relativistic). The procedure just described is intrinsic and, therefore, besides of its simplicity, it has the advantage of being absolutely general; therefore, it is valid for configuration spaces equipped with an arbitrary pseudo-Riemannian metric. Anyway, it is much more important to point out the following: it is not possible to understand a relativistic fluid as an aggregate of relativistic particles (and then a limit process). The reason is that we cannot even formulate mechanical systems defined by several relativistic particles (see the preprint \cite{Objection} joint Muñoz-Díaz, where this question is rigorously stated and proved). However, quite surprisingly, from the mechanical system describing a single relativistic particle, by taking each of its intermediate integrals, we get the relativistic fluid model.
\smallskip

The aforementioned time constraint (introduced by Muñoz-Díaz in \cite{MecanicaMunoz}) and the relativistic correction (introduced in \cite{RM}) are univocally and intrinsically defined processes that will be explained in the article.
\bigskip

The structure of the paper is as follows. In Section \ref{s:mecanica}, we explain the classical mechanics approach that will be used throughout the work. The essential feature is to be founded on basic structures from the configuration space $M$: 1) the concept of second order differential equation, which is an object of the tangent bundle of $M$, 2) the symplectic structure of the cotangent bundle of $M$ and 3) the isomorphism that a pseudo-Riemannian metric defines between the tangent and cotangent bundles of M.
Section \ref{s:relatividad} is dedicated to relativistic systems: the solutions of a mechanical system are, by definition, parameterized curves; we will understand by relativistic systems those whose solutions are parameterized precisely by the arc length associated with the metric (Muñoz-Díaz introduced this  definition in \cite{RelatividadMunoz}; see also \cite{RM}).
Section \ref{s:ligaduras} deals with the aforementioned constraints. The first is the \emph{time constraint}: there is no function on the tangent fiber that can play the role of time; however, it is always possible to make a constraint that forces a mechanical system, by means of a canonical modification, so that a given function can play that role. The second one is the \emph{relativistic correction}: given a mechanical system, it is possible to canonically modify the form of force in the direction of its trajectories so that they are parameterized by the arc length (\emph{proper time}), thus becoming relativistic. Intermediate integrals are introduced in Section \ref{s:intermedias}: some basic properties, examples and a criterion characterizing them will be given.
In Section \ref{s:fluidoclasico} it is shown that Euler's equations are identical to those of the intermediate integrals of a system consisting of a single point particle under the influence of an integrable force. The most interesting case is that of unsteady fluids that are reached after a time constraint; in that situation, Euler equation comes from taking the ``spatial part'' of the equation of the intermediate integrals, while the ``time component'' gives the Bernoulli equation, in its unsteady version. Two cases are developed in detail: 1) when the spatial metric is static and 2) when the spatial metric varies with time giving rise to a FLRW type global metric. The relativistic case is given in Section \ref{s:fluidorelativista}: It is shown that Euler's relativistic equations for a perfect fluid are obtained as the equations of the intermediate integrals of a mechanical system given by a point particle under the action of an integrable force once it is corrected relativistically. Lastly, Section \ref{s:conclusiones} just contains some final remarks and conclusions.
\bigskip

\section{Mechanical systems}\label{s:mecanica}
We briefly present in this section the bases of the Mechanics that will be used in the following sections. It is an approach given in \cite{MecanicaMunoz} and subsequently developed and applied in several directions \cite{RM,RAlonso,Flux,RelatividadMunoz,Tiempo}. The main idea is to present the theory as based on the most essential and geometric facts of the second order equations. It will also help us to fix the notation. In Section \ref{s:relatividad}, we give the definition of relativistic mechanical systems: those whose trajectories are parameterized by the length element; Surprisingly enough, it will turn out that such property does not depend on the metric.

Let $M$ be a smooth manifold of dimension $n$, and $TM$ be its
tangent bundle. From here on, we will denote by $x^1,\dots,x^n$ a collection of local coordinates in (some open set of) M.

Each differential $1$-form $\alpha$ on $M$ can be
considered as a function on $TM$, denoted by $\dot\alpha$, which
assigns to each $v_a\in T_a M$ the value
\begin{equation}\label{e:puntoforma}
\dot\alpha(v_a)=\langle\alpha_a, v_a\rangle
\end{equation} obtained by duality.
In particular, a function $f\in\cC^\infty(M)$ defines the function
on $TM$ associated to $df$ that we denote in short by $\dot f$.
This definition also applies to differential forms $\alpha$ on
$TM$ that are at each point the pull-back of a form on $M$. In the
sequel we call these forms {\it horizontal forms} (locally they look like $\alpha=\alpha_idx^i$, for certain functions on the tangent bundle $\alpha_i$).

From given coordinates $x^i$,  the collection $x^i,\dot x^i$ defines a set of local coordinates for the
corresponding open set of $TM$.   In these type of coordinates, if $\alpha=\alpha_i(x,\dot x)\, dx^i$ is a horizontal 1-form then
\begin{equation}\label{e:puntoforma2}
\dot\alpha=\alpha_i\,\dot x^i.
\end{equation}
In particular, for a function $f$,
\begin{equation}\label{e:puntofuncion}
\dot f=\dot x^i\frac{\partial f}{\partial x^i}.
\end{equation}

The map $f\mapsto \dot f$  from $\cC^\infty(M)$ to
$\cC^\infty(TM)$ is a derivation of the ring $\cC^\infty(M)$
taking values in the $\cC^\infty(M)$-module $\cC^\infty(TM)$. We
denote it by $\dot d$ since it is essentially the differential. In coordinates,
\begin{equation}\label{dpunto}
\dot d=\dot x^i\frac{\partial }{\partial x^i}.
\end{equation}
For each horizontal form $\a$, we have $\dot\a=\la\a, \dot d\ra$
as functions on $TM$ (pairing as usual).
\medskip

If we consider a parameterized smooth curve $\gamma\colon I\to M$, 
($I$ is a open real interval), we can consider the set of its velocities
$$v_{\gamma(t)}:=\gamma_*\left(\frac{d}{dt}\right)_t\in T_{\gamma(t)}M,\quad t\in I.$$
This define another parameterized curve, say $\dot\gamma\colon I\to TM$, which is called the \emph{lift} or \emph{prolongation} of $\gamma$ to $TM$. In local coordinates, if $\gamma$ is described by $x^i=x^i(t)$ then $\dot\gamma$ is
$$x^i=x^i(t),\quad \dot x^i=\frac{dx^i(t)}{dt},$$
so that the ``dot'' notation is consistent with its customary meaning.

\medskip

\subsection{Second order differential equations. Contact system}

In this context, we will deal with  a geometric version of the concept of system of explicit second order ordinary differential equations:

\begin{defi}
\label{def: second_order_equation}
A tangent vector field $D$ on $TM$ is a \emph{second order differential
equation} when its restriction (as derivation) to the subring
$\cC^\infty(M)$ of $\cC^\infty(TM)$ is $\dot d$. This is
equivalent to have $\pi_*(D_{v_a})=v_a$ at each point $v_a\in
T_aM$ (where $\pi\colon TM\to M$ denotes the canonical
projection and $\pi_*$ is its associated tangent map).
\end{defi}

In local coordinates $x^i,\dot x^i$, in $TM$,
a second order differential equation has the expression
\begin{equation*}
D=\dot x^i\frac{\partial}{\partial x^i}+f^i(x,\dot x)\frac{\partial}{\partial \dot x^i}.
\end{equation*}
Usually we will denote $f^i$ by $\ddot x^i$ understanding that it
is a given function of the $x$'s and $\dot x$'s.

The difference of two second order differential equations is a vector field $TM$ vertical with respect to the projection $TM\to M$,  so that annihilates all the functions of $M$. Locally a vertical vector field looks like
$$V=V^i\frac{\partial}{\partial \dot x^i}.$$
Conversely, the sum of a second order differential equation and a vertical vector field is a new second order differential equation. Therefore, we can say that the set of second order differential equations  has an affine structure modeled on the vertical vector fields.

\begin{defi}
\label{def: Pfaff_system}
The \emph{contact system} on $TM$ is the Pfaff system  which
consists of all the $1$-forms annihilating all the second order
differential equations. It will be denoted by $\Omega$.
\end{defi}

The forms in the contact system also annihilate the
differences between pairs of second order differential equations, so that annihilate all the vertical vector fields.
Therefore, they are horizontal forms; each
$\omega_{v_a}\in\Omega_{v_a}$ is the pull-back to $T^*_{v_a}TM$ of
a form in $T^*_aM$. Now, a horizontal $1$-form kills a second
order differential equation if and only if it kills the field
$\dot d$. Thus the contact system on $TM$ consists of the
horizontal $1$-forms which annihilate $\dot d$: in other words, a horizontal form $\alpha$ is contact if and only if $\dot\alpha=0$.
\medskip

A local system of generators for the contact system $\Omega$, out
of the zero section, is given by
\begin{equation*}
\dot x^idx^j-\dot x^jdx^i\quad (i,j=1,\dots,n)\ .
\end{equation*}

 Finally, the contact system has the following interpretation: the one-dimensional solutions of $\Omega$ are the liftings to $TM$ of parameterized curves in $M$ (provided that we consider solutions out of the zero section which are regularly projected to $M$).
\medskip

\subsection{Cotangent bundle. Liouville form}
Let $T^*M$ be the cotangent bundle of $M$ and $\pi\colon T^*M\to M$
the canonical projection. Each set of local coordinates $x^i$ on $\mathcal{U}\subseteq M$, induce local coordinates
$x^i,p_i$ on $\pi^{-1}(\mathcal{U})\subseteq T^*M$, as follows: let $\alpha_a\in T^*M$, then, by construction, $x^i(\alpha_a):=x^i(a)$ and $p_i(\alpha_a):=\alpha_a(\partial/\partial x^i)_a$.
\medskip

Recall that the \emph{Liouville form}
$\theta$ on $T^*M$ is defined by the rule
$$\theta_{\alpha_a}=\pi^*(\alpha_a)$$
for $\alpha_a\in T^*_aM$. Abusing the notation we can write
$\theta_{\alpha_a}=\alpha_a$.
\medskip

The $2$-form $\omega=d\theta$ is the natural symplectic form
associated to $T^* M$. In local coordinates  we have
\begin{equation*}
\theta=p_idx^i,\qquad \omega=dp_i\wedge dx^i .
\end{equation*}
\medskip

\subsection{Pseudo-Riemannian manifolds. Kinetic energy}
Let $g$ be a pseudo-Riemannian metric in $M$.
Then we have an isomorphism of vector fiber bundles
\begin{align*}\label{bemol}
TM &\to T^*M\\
v_a &\mapsto v_a^\bemol\qquad\text{where}\quad v_a^\bemol:=\iiota_{\displaystyle{v_a}}g
\end{align*}
($a$ is a point of $M$ and $\iiota_{v_a}g$ is the inner contraction of $v_a$ with $g$, so that $v_a^\bemol(e_a)=g(v_a,e_a)$ for all vector $e_a$).
In local coordinates, the above isomorphism is
$$x^i=x^i,\quad p_i=g_{ij}\dot x^j,$$
or, also,
$$v_a=v_a^i\left(\frac{\partial}{\partial x^i}\right)_a\quad\longmapsto\quad v_a^\bemol=g_{ij}(a)v_a^j\,d_ax^i.$$

Using
the above isomorphism we can transport to $TM$ all structures on
$T^*M$. In particular, we work with the Liouville form $\theta$
and the symplectic form $\omega$ transported in $TM$ with the
same notation.

From the definitions we have for the Liouville form in $TM$, at
each $v_a\in T_aM$,
\begin{equation}\label{Liouville}
 \theta_{v_a}=v_a^\bemol,
\end{equation}
where the form of the right hand side is to be understood
pulled-back from $M$ to $TM$.

\begin{defi}
\label{cinetica}
The function $T:=\frac 12\,\dot\theta$ on $TM$ (see (\ref{e:puntoforma})) is the \emph{kinetic
energy} associated to the metric $g$. So, for each $v_a\in TM$,
we have $$T(v_a)=\frac 12 \,\dot\theta(v_a)=\frac 12 g(v_a,v_a).$$
\end{defi}
In local coordinates, as usual,
$T(v_a)=\frac 12 g_{ij}(a)v_a^iv_a^j,$
where $v_a=v_a^i(\partial/\partial x^i)_a$; equivalently (see Equation (\ref{e:puntoforma2})),
\begin{equation}\label{e:cinetica}
T=\frac 12\, g_{ij}\dot x^i\dot x^j=\frac 12\, p_j \dot x^j. 
\end{equation}

\medskip

\subsection{Newton's second law}
Now we will address the core of the theory: the laws of motion.
When  $M$ is endowed with a pseudo-Riemannian metric $g$, we get a canonically associated second order differential equation: the so called \emph{geodesic field}. The first law of the Mechanics says that, if no force is impressed, the trajectories of a mechanical system follow the geodesics of $g$ (\emph{Inertia Law}):
\begin{defi}
\label{campogeodesico}
The \emph{geodesic field}  $D_G$ is the vector field on $TM$ univocally determined by
\begin{equation}\label{ecuaciongeodesica}
\iiota_{D_G}\omega+dT=0 \ .
\end{equation}
In other terms, given $g$, the geodesic field $D_G$ is the \emph{symplectic gradient} of (minus) the kinetic energy.
\end{defi}
Observe that $D_GT=0$, as it follows by contracting (\ref{ecuaciongeodesica}) with $D_G$.

In local coordinates $x^i,\dot x^i$,
\begin{equation}\label{geodesico}
D_G=\dot x^i\frac{\partial}{\partial x^i}-\Gamma_{ij}^k \dot x^i
\dot x^j \frac{\partial}{\partial \dot x^k}
\end{equation}
The interesting thing is that
\begin{prop}
The geodesic field is a second order differential equation.
\end{prop}
The solutions of $D_G$ are the geodesic curves of $g$.
\bigskip

Now we prescribe the central point of Newtonian Mechanics which simply states the following:
\medskip

\emph{``A mechanical system evolves according to a second order differential equation.''}
\medskip

 On the other hand, we can obtain all of the second order equations by adding a vertical tangent vector field $V$ (=\emph{acceleration}) to the geodesic field (=inertial motion). This vertical vector field translates by means of the symplectic structure $\omega$ into a horizontal differential 1-form $\alpha$ (=\emph{force form}) and conversely:
 \begin{equation}\label{fuerzaaceleracion}
 \iota_V\omega=\alpha,\quad\text{or, in coordinates,}\quad V^i=\alpha^i
 \end{equation}
 where $V=V^i\partial/\partial \dot x^i=V_i\partial/\partial p_i$, $\alpha=\alpha_i dx^i$ and we raise and lower  the indexes as usual: $V_i:=g_{ij}V^j$, $\alpha^i:=g^{ij}\alpha_j$.
   In other words, thanks to the metric $g$, accelerations and forces are bi-univocally determined.
 \medskip

Then, from the above considerations, it is enough to take $D=D_G+V$ to prove the following main result:
\begin{thm}\textbf{\em (Newton's Law, \cite{MecanicaMunoz})}\label{teoremaalfa}
The metric $g$ establishes a one-to-one correspondence between
second order differential equations and horizontal $1$-forms in
$TM$.

The second order differential equation $D$ and the horizontal
$1$-form $\a$ that correspond to each other are related by
\begin{equation}\label{eq:Newton}
{\emph \iiota_D\omega+dT+\alpha=0}.
\end{equation}
The triple $(M,g, \alpha )$ will be called a \emph{mechanical
system}:
``Under the influence of a force $\alpha$ the
trajectories of the mechanical system satisfy (= are the integral
curves of) the second order differential equation $D$ associated with $\alpha$''.
\end{thm}

In local coordinates,
let the second order differential equation $D$ and the differential 1-form $\alpha$ be given respectively by
\begin{equation}\label{formulaD}
D=\dot x^i\frac{\partial}{\partial x^i}+\ddot x^i
\frac{\partial}{\partial \dot x^i} \ ,\qquad \alpha=\alpha_i\,dx^i,
\end{equation}
where the $\ddot x^i$'s, and the $\alpha_k$'s are certain functions of $x,\dot x$. Then, a straightforward computation gives explicitly the relationship
\begin{equation}\label{alfa2}
\ddot x^k=-(\alpha^k+\Gamma_{ij}^k\dot x^i\dot
x^j).
\end{equation}
\medskip

\section{Relativistic mechanical systems}\label{s:relatividad}
We will consider relativistic systems as particular cases of classical mechanical systems within the previous approach. Let us describe now what is the peculiarity of such systems: motion of a particle in Relativity
(yet in the special one) is parameterized by the ``proper time'' of
that particle, whose ``infinitesimal element'' $ds$ is the so
called length element associated with  a Lorentzian metric $g$.
\medskip

For any pseudo-Riemannian metric $g$, the length element along a trajectory $x^i=x^i(t)$ parameterized by a parameter $t$ is usually written as
$$ds=\sqrt{\left|g_{ij}\frac{dx^i}{dt}\frac{dx^j}{dt}\right|}\,dt.$$
We look for a well defined differential 1-form on $TM$, which coincides with $ds$ along each trajectory (once lifted from $M$ to $TM$). The expression under the square symbol is easily recognizable as the specialization of $\dot\theta=2T$ (see (\ref{e:cinetica})).

The main task now is to find a ``generic'' substitute for ``$dt$'' (recall that $t$ is just the parameter for one single trajectory): any quotient $\beta/\dot\beta$ for a given horizontal  differential 1-form $\beta$  can play that role; we take $\beta=\theta$ in order to get a canonical choice. Finally, we define the \emph{length element} (maintaining the traditional notation) to be
$$ds:=\frac{\theta}{\sqrt{|\dot\theta|}}=\frac{\theta}{\sqrt{|2T|}},$$
where
$T$ is the kinetic energy function (this intrinsic construction of the length element as a differential 1-form defined on $TM$ was done in \cite{RelatividadMunoz} (see, also, \cite{RM})). In local coordinates:
$$ds=\frac{p_i}{\sqrt{|2T|}}\,dx^i=\frac{p_i}{\sqrt{|p_j\dot x^j|}}\,dx^i=\frac{p_i}{\sqrt{|g^{jk}p_jp_k|}}\,dx^i.$$

In this way,  the
classical ``$ds$'' is the restriction of the length element $ds$
to the curve in $TM$ describing the lifting of the corresponding
parametrized curve in $M$. Then we can integrate $ds$ along the parameterized trajectory (lifted to $TM$) to obtain ``a'' length. It turns out that this length does not depend on the parametrization, and this is the reason for which it makes sense  to talk about the length of a curve.

If $D$ is a second order differential equation on $M$, when we say
that a curve solution can be parametrized  by the length element
$ds$ we means that the proper parameter for such a curve solution
of $D$ is the specialization of $ds=\theta/\sqrt{|\dot\theta|}$ (we assume $\dot\theta\ne 0$); that
is to say,
$$\theta(D)/\sqrt{|\dot\theta|}=\dot\theta/\sqrt{|\dot\theta|}=1.$$
Thus, $$\dot\theta=2T=\pm 1$$ on
such a curve solution of $D$.

Therefore, the second order differential equations on $M$ which
describe relativistic motions are vector fields $D$ (on
$TM$) tangent to the hypersurfaces $\dot\theta=2T=\pm1$. As a consequence, $DT=0$ on $T=\pm 1/2$.

On the other hand, by inner contraction with $D$ of the
Newton's equation we get
$$0=\iiota_D\iiota_D\omega+\iiota_D dT+\iiota_D\alpha=DT+\dot\alpha,$$
since $\iiota_D\alpha=\dot\alpha$ because $\alpha$ is horizontal (only terms $dx^i$ appear and $\iiota_Ddx^i=Dx^i=\dot x^i$).

Let us put $D=D_G+W$. From $D_GT=0$, we obtain
\begin{equation}\label{EcuacionT}
WT+\dot\alpha=0.
\end{equation}
 On the other hand, locally, $W=W_i\partial/\partial p_i$ and $T=(1/2)g^{ij}p_ip_j$ so that
$$WT=g^{ij}p_iW_j=\dot x^jW_j.$$
 In this way, in the particular case when the $W_j$ are homogeneous functions with respect to $\dot x^i$, the condition $WT=0$ for $\dot\theta=2T=1$ implies
 $WT=0$ for an arbitrary value of $T$; that is to say, $WT$ (and so $DT$) always vanishes for that class of $W$.

For all the above, it seems reasonable to give the following
\begin{defi}\label{DefRelativista}
A \emph{relativistic field} on $(M,g)$ is a second order
differential equation $D$ such that $DT=0$, where $T$ is the
kinetic energy function associated with $g$.
\end{defi}

The previous discussion gives us
\begin{thm}\label{criterio}
A second order differential equation $D$ on $(M,g)$ is a
relativistic field if and only if the associated 1-form of force $\alpha$  meets one (and therefore both) of the following two equivalent conditions
\begin{enumerate}
\item $\dot\alpha=0,$
\item $\alpha$ is a contact form.
\end{enumerate}
\end{thm}

It is very important to remark that the above condition results to be independent of the metric $g$; in particular, it has nothing to do with the signature of $g$.
\medskip

In coordinates, a force $\alpha$ is relativistic if and only if
$$\alpha_i\,\dot x^i=0,$$
which can be read as follows: relativistic forces are those that does not work (\emph{gyroscopic forces}).

A significant class of relativistic mechanical systems are, of course, those defined by Lorentz type forces:
$$\alpha:=\iiota_{\dot d}F=\dot x^iF_{ij}\,dx^j,$$
where $F$ is a differential 2-form $F=F_{ij}dx^i\wedge dx^j$ so that
$$\dot\alpha=\dot x^iF_{ij}\dot x^j=0$$
because $F_{ij}=-F_{ji}.$
In the case of electromagnetic forces, the forms $F$ come from the configuration manifold $M$, but nothing prevents them, even if they are horizontal, depending on the velocities:
$F_{ij}=F_{ij}(x,\dot x)$; as a particular type, it could be $F_{ij}=F_{ij,k}(x)\dot x^k$. In that case, the force form depends quadratically on the velocities, etc.
\medskip

On the other hand, a \emph{conservative system} $\alpha=d\Phi$ (for some function $\Phi$ in $M$) can never be relativistic  because $\dot\alpha=\dot\Phi\ne 0$ except in the trivial case $\alpha=0$ (the \emph{free of forces system}). Observe that this conclusion is, a priori, completely independent of any consideration about ``action at a distance''.
\medskip

We refer to \cite{RelatividadMunoz} for a more detailed analysis of relativistic forces in the sense above defined.
\bigskip

\section{Constraints}\label{s:ligaduras}

The general presentation of mechanical systems with constraints in the approach given in Section \ref{s:mecanica} can be seen in \cite{MecanicaMunoz,RM}. Here we will limit ourselves to developing two particular cases that will be used later: time constraint and relativistic correction.

\subsection{Time constraint}
In a mechanical system build on a  pseudo-Rie\-man\-nian manifold as above we do not have a priori a time function (which is not possible, although this question admits a alternative approach; see the remark below). However, it is always possible to force a given function $t$ to play the role of time, canonically modifying the mechanical system. This constraint will ``produce'' a force on the trajectories of the original system making $\dot t$ to stay constant.

In order to put at its true value the process of time constraint, we must note the following remarkable  fact (see \cite{RM,Tiempo}): every conservative $n$-dimensional mechanical system can be obtained by applying a time constraint on an appropriate $(n+1)$-dimensional inertial system (inertial in the sense of being ``free of forces'': $\alpha = 0$ in Equation (\ref{eq:Newton})). This scheme is reminiscent  of the Kaluza-Klein theory in some respects.

\begin{obs}
There is a concept that replaces the (non existing!) function ``time'' on the tangent bundle; this was introduced by Muñoz-Díaz in \cite{MecanicaMunoz} and was called the \emph{class of time}: it is a well-defined family of differential forms that measures duration along each trajectory; the existence of the class of time rests exclusively on the infinitesimal structure of the configuration manifold and is, therefore, independent of the datum of a particular pseudo-Riemannian metric. See also \cite{Tiempo} where this concept is discussed in detail giving a clear meaning to the so-called \emph{absolute time} and, in addition, possible approaches to the role of time in Quantum Mechanics are proposed.
\end{obs}

We now will describe how the time constraint is performed.
We start from an arbitrary  mechanical system $(M,g,\alpha)$ with Newton equation
\begin{equation}\label{Newtont1}
\iiota_D\omega+dT+\alpha=0\ .
\end{equation}

Now, let us fix a function $t\in\C^\infty(M)$ so that $dt$ is regular (eventually, shrinking the domain of validity of the equations). We want trajectories where $\dot t$ keeps constant. In other words, $\dot t$ have to be a first integral of the equations when suitably modified.

With the above purposes we consider the following congruence:
\begin{align}
& \iiota_{\overline D}\,\omega+dT+\alpha+\lambda\, dt=0\label{Newtont2} \\
& \overline D\dot t=0,
        \label{Newtont3}
\end{align}
for a suitable ``multiplier'' $\lambda\in\C^\infty(TM)$.

In order to find the second order differential equation $\overline D$, we put $\overline D=D+W$, where $D$ solves  (\ref{Newtont1}) and $W$ is a vertical vector field on $TM$. In local coordinates $\{t,x^1\dots,x^n,p_0,p_1,\dots,p_n\}$ (say, $t=x^0$), we have $W=W_\alpha\partial/\partial p_\alpha$, $\alpha=0,\dots,n$ and (\ref{Newtont2})+(\ref{Newtont3}) become
$$\iiota_W\omega+\lambda dt=0,\quad W\dot t=-D\dot t.$$
\smallskip

On the other hand, $\iiota_W\omega=W_\alpha\,dx^\alpha$ and $\dot t=g^{0\alpha}p_\alpha$, so that
$$W_0+\lambda=0,\,\, W_i=0, i\ge 1,\quad\text{and}\quad g^{00}W_0=-D\dot t.$$
Finally we get, locally,
$$W=-\frac{D\dot t}{g^{00}}\frac{\partial}{\partial p_0}.$$
This proves that $W$ (and so, also  $\overline D$)  exists and is uniquely determined by:
$$\iiota_{\overline D}\,\omega+dT+\alpha-\frac {D\dot t}{g^{00}}\,dt=0,$$
which is equivalent to substitute in Equation (\ref{Newtont1}) the force form $\alpha$ by the modified force form
$$\overline\alpha:=\alpha+\lambda dt$$
where the ``multiplier'' $\lambda$ equals $-(D\dot t)/g^{00}$.
(For a more intrinsic proof of the existence of $\overline D$, see \cite{Tiempo}). 

\medskip

By it very construction, the curve solutions of $\overline D$ are parameterized by constant multiples of $t$.
\medskip

\subsection{Relativistic correction}

In this case, as we shall see, the constraint has a slightly different nature. As before, we start from an arbitrary mechanical system  (\ref{Newtont1}); as a rule, that system is not relativistic in the sense of Definition \ref{DefRelativista} because $\dot\alpha$ can be non-null (Theorem \ref{criterio}); this situation includes all the conservative systems $\alpha=d\Phi$ except the geodesic case $d\Phi=0$.
\medskip

The modification to achieve relativistic trajectories (in the sense of the Definition \ref{DefRelativista}) will consist in exerting a force in the very trajectory's direction. Generically, such a direction is given by $\dot d=\dot x^i\partial/\partial x^i$ (Section \ref{s:mecanica}, Equation (\ref{dpunto})); on the other hand, $\iota_{\dot d}\,\omega=\theta$, so that, finally, what we will do is consider the addition of forms of force that are multiples of $\theta$.

In view of the above discussion, in order to get a relativistic system we can modify the force along the trajectories as follows: we replace $\alpha$ (in the open set $\dot\theta\ne 0$) by
$$\widehat\alpha:=\alpha-\frac{\dot\alpha}{\dot\theta}\,\theta,$$
which, by construction, is a relativistic force:
$$\dot{\widehat\alpha}=\dot\alpha-\frac{\dot\alpha}{\dot\theta}\,\dot\theta=0.$$
In other terms, we substitute (\ref{Newtont1}) by
\begin{equation}\label{NewtonR}
\iiota_D\omega+dT+\alpha-\frac{\dot\alpha}{\dot\theta}\,\theta=0\ .
\end{equation}

  From the point of view of general Mechanics, the above procedure is the natural \emph{``relativistic correction''} of (\ref{Newtont1}).
  \medskip

  As a particular instance, it can be seen that Formula (2.25), Ch. II, of \cite{Barut} is obtained as the result of the above proposed relativistic correction: when the initial force $\alpha$ is conservative, $\alpha=d\Phi$, we have
  $$\widehat\alpha=d\Phi-\frac{\dot\Phi}{\dot\theta}\,\theta,$$ which, in local coordinates, is
  $$\widehat\alpha_i=\Phi,_{i}-\frac{\Phi,_{j}\dot x^j}{\dot\theta}\,p_i,\qquad\text{or}\qquad
    \widehat\alpha^k=\Phi,_{i}\left(g^{ik}-u^ku^i\right),$$
    where $\Phi,_i:=\partial\Phi/\partial x^i$ and $u^i:=\dot x^i/\sqrt{\dot\theta}$ (for concreteness, we assume $\dot\theta>0$).
 \medskip

 Following the same steps, we can obtain the coordinate expression for $\widehat\alpha$ in general:
 $$\widehat\alpha_i=\alpha_i-\alpha_ju^ju_i,\qquad\text{or}\qquad \widehat\alpha^k=\alpha_i\left(g^{ik}-u^ku^i\right).$$

\bigskip

\section{Intermediate integrals}\label{s:intermedias}

This is a classical concept in the theory of differential equations: an intermediate integral of a given differential equation is another differential equation of lower order, whose solutions are also solutions of the above one. As an elementary example, first order differential equations $y'=c$ are intermediate integrals of the second order equation $y''=0$. In some situations, to integrate all or a part of the intermediate integrals can lead to integrate the starting equation. This is the case of the well known Hamilton-Jacobi equation (see below).

We now  are interested in intermediate integrals of a second order differential equation $D$ (Definition \ref{def: second_order_equation}); therefore, we look for vector fields $v$ on $M$ (first order differential equations) whose curve solutions are also solutions of $D$. Note that solutions of $v$ are parameterized curves in $M$, so we must prolong these curves to $TM$ in order to be able to say whether they are solutions of $D$ or not. What happens is that, due to its special structure, all solutions of $D$ as curves in $TM$ are prolongation of curves that come from $M$.

\begin{lemma}
The tangent fields $v$ on $M$ that are intermediate integrals of $D$ in (\ref{eq:Newton}) are precisely those holding
\begin{equation}\label{NewtonIntermedia}
\iiota_{\displaystyle v}\, dv^\bemol+dT(v)+v^*\alpha=0,
\end{equation}
where $v^*\alpha$ is the pull-back of $\alpha$ by means of the section $v\colon M\to TM$ and $T(v)=v^*T$ is the kinetic energy function $T$ specialized to the image of $v$.
\end{lemma}
The essence of the demonstration of the above lemma, which we will not give here, is to apply that $v$ is an intermediate integral of $D$ if and only if it is fulfilled that $D_{v_x}=v_*v_x$ for all $x\in M$; see \cite{RM}.
\medskip

The identity stated in the following lemma is a modern version of the classical formula of $\R^3$ vector calculus ``$\textrm{grad}(\|v\|^2/2)=v\times \textrm{curl}\,{v}+(v\,\cdot\textrm{grad})\,v$''. A proof can be found, for instance, in \cite{Flux}, Lemma 3.4.
\begin{lemma}
With the above notation
\begin{equation}\label{identidadvorticidad}
\iiota_{\displaystyle v}\, dv^\bemol+dT(v)=\left(v^\nabla v\right)^\bemol.
\end{equation}
\end{lemma}
In particular, equation for the intermediate integrals (\ref{NewtonIntermedia}) becomes
\begin{equation}\label{NewtonIntermedia2}
\left(v^\nabla v\right)^\bemol+v^*\alpha=0.
\end{equation}
In the case when $\alpha$ is a differential form on $M$ and if we denote by $\textrm{grad}\,\alpha$ the vector field such that
$$(\textrm{grad}\,\alpha)^\bemol=\alpha,$$
then (\ref{NewtonIntermedia}) is equivalent to
\begin{equation}\label{NewtonIntermedia3}
v^\nabla v+\textrm{grad}\,\alpha=0.
\end{equation}
 That is the version for intermediate integrals of the Newton equation in the form ``\emph{mass $\times$  acceleration = force}''. A more general version, when $\alpha$ depends on the velocities (a horizontal differential form on $TM$) is
\begin{equation}\label{NewtonIntermedia4}
v^\nabla v+\textrm{grad}\,v^*\alpha=0.
\end{equation}
\medskip

\subsection{Conservative systems. Hamilton-Jacobi equation}\label{sconservativos}

As it is well known, when $\alpha$ is an exact differential 1-form, the system
$(M,g,\alpha)$ is said to be  \emph{conservative}. Since
$\alpha$ is horizontal, the potential function $U$ such that
$\alpha=dU$ it has to belong to $\cC^\infty(M)$. The function  $H=T+U$ is the so-called \emph{Hamiltonian} function of the mentioned system. In this case, the Newton equation (\ref{eq:Newton}) is
\begin{equation}\label{newtonconservada}
\iiota_D\omega+dH=0,
\end{equation}
or Hamilton’s Canonical Equations of Motion.

Equation (\ref{NewtonIntermedia}) for the intermediate integrals of $D$
is, now:
\begin{equation}\label{intermediaconservada}
\iiota_{\displaystyle v}dv^\bemol+dH(v)=0
\end{equation}
In particular, when $v$ is a Lagrangian submanifold $TM$ (or equivalently, when $v^\bemol$ is a Lagrangian submanifold of $T^*M$), it holds, by definition,
$dv^\bemol=d\theta|_v=0$, and the above equation is simply
$$dH(v)=0\quad\text{ó}\quad H(v)=\textrm{cte.}$$
Locally, Lagrangian condition means that exists an smooth function $S$ such that $v=\textrm{grad}\,S$ or $v^\bemol=dS$: in this way we get
\begin{prop}
The \emph{Hamilton-Jacobi equation} for a conservative mechanical system (\ref{newtonconservada}) is the partial differential equation of its Lagrangian intermediate integrals
\begin{equation}\label{HamiltonJacobi}
H(\textrm{grad}\,S)=\textrm{cte.}\,\,\,\text{in $TM$, or}\quad
H(dS)=\textrm{cte.}\,\,\,\text{in $T^*M$}
\end{equation}
\end{prop}
This is a conceptual characterization of the Hamilton-Jacobi equation from which, as it is well known, all the structure of the mechanical system is recovered.
\bigskip

\section{Classical fluids as intermediate integrals}\label{s:fluidoclasico}
We will see in this section how Euler equations of fluids are recovered as those of intermediate integrals of certain classical mechanical systems. When the fluid is isentropic, it is enough to start from a conservative mechanical system. In more general cases we must start with systems with integrable forces in the sense that they admit an integrating factor. In order to stagger the exposition, we start from the simplest case of steady fluids and then address the time-dependent case.

\subsection{Steady case} This is an almost obvious result from the formal point of view. Let us consider a mechanical system $(M,g,\alpha)$ in the case
$$\alpha=\frac{dP}\rho,$$
for given arbitrary functions $P$, $\rho\in\C^\infty(M)$. In that case the equation of the intermediate integrals (\ref{NewtonIntermedia3}) is
$$v^\nabla v+\frac{\textrm{grad}P}{\rho}=0.$$
We recognize immediately Euler equation where $P$ is the pressure and $\rho$ the density.

When $P$ and $\rho$ are functionally dependent, there is a function $\Phi$ such that
$\alpha=dP/\rho=d\Phi,$
and the equation reduces to
$$v^\nabla v+\textrm{grad}\,\Phi=0.$$

If, for example, $M=\R^3$ and $g$ the usual euclidean metric, then we can say that the fluids with $d\Phi=dP/\rho$ are, at the same time, the intermediate integrals of  the equation describing the motion of a point particle under the action of a potential force $d\Phi$.
\medskip

There is no need to say that we can consider also the addition of another force form, say $\alpha=dP/\rho+f$, and we get a quite general Euler equation for an steady ideal fluid:
$$v^\nabla v+\frac{\textrm{grad}P}{\rho}+{\textrm{grad}\,f}=0.$$

\subsection{Unsteady case}
In this particular case, we will consider a manifold $M:=\R\times M^{{s}}$, endowed with a pseudo-Riemannian metric
\begin{equation}\label{metricat}
g=dt^2+ g^s
\end{equation}
where $t$ is the coordinate of $\R$ and $g^s$ is a family of pseudo-Riemannian metrics on $M^s$ (lifted to $M$ by pull back), parameterized by $t$. The idea is to consider $t$ as the time, and $g^s$ as a ``spatial metric''. That situation is not too restrictive: this is the case of the classical mechanics and also that of the relativistic Friedmann-Lemaître-Robertson-Walker cosmological models (v.g., \cite{Hawking}, pp. 134 ff); In fact, locally it is not a restriction at all, although it may involve a fairly arbitrary choice: the use of geodesic gaussian coordinates associated to an appropriate ``spatial initial hypersurface'' puts, locally, a pseudo-Riemannian metric in the above form up to a sign: along the hypersurface, it is chosen a normal vector field which is taken as the set of initial conditions of geodesic curves; in this way, the length of arc of these geodesics plays the role of $t$ (see, for instance, \cite{Eisenhart}, p. 57 and, for the Lorentzian signature case, \cite{Adler}, pp. 59 ff).

 Now, let us consider on $M$ a mechanical system
 \begin{equation}
 \iiota_D\,\omega+dT+\alpha=0,\quad\text{where}\quad \alpha=\frac{dP}\rho,
 \end{equation}
 for given functions $P$ and $\rho$, i.e., we consider an integrable force 1-form $\alpha$, where the function $\rho$ is an integrating factor.

 Now, let us perform on the above mechanical system a time constraint $\dot t=\text{constant}$; the new system  is (\ref{Newtont2})+(\ref{Newtont3}):
 \begin{equation}\label{Newtont4}
 \begin{cases}
 \iiota_{\overline D} \omega+dT+\alpha+\lambda\,dt=0\medskip\\
 \overline D\dot t=0
  \end{cases}
 \end{equation}
 By contracting the first equation of (\ref{Newtont4}) with $\overline D$ we get $(dT+\alpha)(\overline D)+\lambda\dot t=0$ which gives us the useful identity
 $$\lambda=-\frac{(dT+\alpha)(\overline D)}{\dot t}.$$

 According to Lemma \ref{NewtonIntermedia} and Equation (\ref{NewtonIntermedia2}), if we put $\overline\alpha:=\alpha+\lambda dt$, the intermediate integrals $\overline v$ of (\ref{Newtont4}) are those such that
 \begin{equation}\label{Intermediat1}
 (\overline v^\nabla\overline v)^\bemol+\overline v^*\,\overline\alpha=0
 \end{equation}
 Therefore, we need compute $\overline v^*\lambda$; first, easily $\overline v^*\dot t=\overline v(t)$ and now, by taking into account that $\overline D_{\overline v_x}=\overline v_* \overline v_x$ because $\overline v$ is an intermediate integral, we get
 \begin{align*}
 \overline v^*\{(dT+\alpha)(\overline D)\}(x)&=((dT+\alpha)(\overline D))_{\overline v_x}=(dT+\alpha)_{\overline v_x}\overline D_{\overline v_x}\\
                      &=(dT(\overline v)+\overline v^*\alpha)_x\overline v_x=((dT(\overline v)+\overline v^*\alpha)(\overline v))(x)
 \end{align*}
 and then,
 $$ \overline v^*\{(dT+\alpha)(\overline D)\}=(dT(\overline v)+\overline v^*\alpha)(\overline v).$$
 \medskip

 In the specific case $\alpha=dP/\rho$ we have $\overline v^*\alpha=dP/\rho$ (as a form on $M$) and then:
 $$ \overline v^*\{(dT+\alpha)(\overline D)\}=\overline v(T(\overline v))+\frac{\overline v(P)}{\rho}.$$
 In this way, (\ref{Intermediat1}) is
 \begin{equation}\label{Intermediat2}
 (\overline v^\nabla\overline v)^\bemol+\frac{dP}\rho-\frac{\overline v(T(\overline v))+\overline v(P)/\rho}{\overline v(t)}\,dt=0
 \end{equation}
 or, in terms of vector fields,
 \begin{equation}\label{Intermediat3}
 \overline v^\nabla\overline v+\frac{\textrm{grad}\,P}\rho-\frac{\overline v(T(\overline v))+\overline v(P)/\rho}{\overline v(t)}\,\textrm{grad}\,t=0.
 \end{equation}
 Since  (\ref{metricat}) it follows that
 $$\textrm{grad}\,t=\frac{\partial}{\partial t}.$$

 \begin{obs}
Although we have considered the constraint $\dot t=$constant, not every intermediate integral $\overline v$ of the thus constrained mechanical system meets that $\overline v(t)$ is constant: for example $\overline v = x\,\partial/\partial t$ is an intermediate integral of $\ddot t = 0$, $\ddot x = 0$, but $\overline v(t)=x$ is not a constant. The only thing we can assure is that $\overline v(t)$ is constant along the trajectories of $\overline v$ or, in other words, $\overline v(t)$ is a first integral of $\overline v$.
\end{obs}

 From now on in this section, we will assume that
 $\overline v(t)=1$. In particular, Equation (\ref{Intermediat3}) becomes
 \begin{equation}\label{Intermediat31}
 \overline v^\nabla\overline v+\frac{\textrm{grad}\,P}\rho-\left({\overline v(T(\overline v))+\frac{\overline v(P)}\rho}\right)\,\frac{\partial}{\partial t}=0.
 \end{equation}

 On the other hand, $\overline v$ can be  decomposed as an orthogonal sum (with respect to the metric (\ref{metricat}))
 \begin{equation}\label{Intermediat4}
 \overline v=\frac{\partial}{\partial t}+v.
 \end{equation}

 If $x^1,\dots,x^n,$ are local coordinates for $M^s$, then $t,x^1,\dots,x^n$ are local coordinates for $M=\R\times M^s$ (in the corresponding open subset); thus,
\begin{equation*}
\overline v=\frac{\partial}{\partial t}+v^i\,\frac{\partial}{\partial x^i};\quad\text{in such a way that}\quad v=v^i\,\frac{\partial}{\partial x^i}.
\end{equation*}
\bigskip

\subsubsection{\bf Static spatial metric}\label{casoestatico}
Now we will assume that the spatial metric $g^s$ does not depend on $t$, so that $g^s$ is a fixed metric on $M^s$.
\medskip

In order to interpret Equation (\ref{Intermediat31}), we will make its orthogonal decomposition (``time and space components'') in the case (\ref{Intermediat4}). Firstly, from the assumption of static $g^s$ we get
$(\partial/\partial t)^\nabla\partial/\partial t=(\partial/\partial t)^\nabla\partial/\partial x^i=0$ so that we have
\begin{equation}\label{vcovariante}
\overline v^\nabla\overline v=\frac{\partial}{\partial t}v+v^\nabla v,
\end{equation}
where $v^\nabla v$ denotes the covariant derivative with respect to the spatial metric $g^s$ ($t$ is just a parameter here) and
$$ \frac{\partial}{\partial t}v:=\frac{\partial v^i}{\partial t}\frac{\partial}{\partial x^i}.$$
In particular, we see that $\overline v^\nabla\overline v$ is purely spatial (has no component in $\partial/\partial t$).
\medskip

In this way, the spatial part of (\ref{Intermediat31}) is
\begin{equation}\label{Intermediat6}
\frac{\partial}{\partial t}v+v^\nabla v+\frac{\textrm{grad}^sP}\rho=0,
\end{equation}
where $\textrm{grad}^s$ denotes the gradient with respect to $g^s$ (coordinate $t$ being a parameter here).
Relation (\ref{Intermediat6}) is a generalized version of the classical  \emph{Euler equation} for a fluid with pressure $P$ and density $\rho$ (generalized just in the sense that we do not necessarily need  $g^s$ to be flat).

With regards to the time part of (\ref{Intermediat31}) we have
$$\frac{1}{\rho}\left(\frac{\partial P}{\partial t}-\overline v(P)\right)-\overline v(T(\overline v))=0;$$
by taking into account that $T(\overline v)=g(\overline v,\overline v)/2=1/2+g^s(v,v)/2$, the above relation can be written as
\begin{equation}\label{Intermediat5}
\frac{v(P)}\rho+v(g^s(v,v)/2)+\frac{\partial g^s(v,v)/2}{\partial t}=0.
\end{equation}
which constitutes an unsteady  version of the \emph{Bernoulli equation}. However, this is not an independent equation because can be derived from
(\ref{Intermediat6}). In fact it is enough to multiply scalarly (\ref{Intermediat6}) by $v$ (with respect to $g^s$) to obtain (\ref{Intermediat5}): for this we must bear in mind that
$$g^s(\textrm{grad}^sP,v)=v(P),\quad g^s(v^\nabla v,v)=\frac 12\,v(g^s(v,v)),\quad\text{and}\quad g^s(\partial v/\partial t,v)=\frac 12\,\frac{\partial g^s(v,v)}{\partial t}.$$
\medskip
As a consequence,
\begin{prop}\label{fluidostiempo}
The intermediate integrals $\overline v$ of a time-constrained mechanical system  (\ref{Newtont4})
for  $\alpha=dP/\rho$, with static metric $g^s$,   such that $\overline v(t)=1$, are the vector fields $\overline v=\partial/\partial t+v$  satisfying
$$ \frac{\partial}{\partial t}v+v^\nabla v+\frac{\textrm{grad}^{\,s}P}\rho=0\quad\text{{(Euler equation).}}$$
\end{prop}

When $g^s$ is flat, the equation above equals to Equation (2.3), p. 3 of \cite{Landau}; in fact, it is easy to obtain more general fluids if we add some force form to $\alpha$ getting Equation (2.4), p.3 \cite{Landau} or Equation (2), 3.43 of \cite{MilneThomson}.
\bigskip

\subsubsection{\bf Friedmann-Lemaître-Robertson-Walker (FLRW) type metric case}\label{casoFLRW}
We will now abandon the static hypothesis. To be able to specify, we will only consider metrics of the form
$$g=dt^2-a(t)^2\,h,$$
where $h$ is a fixed pseudo-Riemannian  metric on $M^s$ and $a(t)$ is a given non-vanishing function of $t$. The usual interpretation of $a(t)$ in FLRW cosmological models is that of a scale factor related to the radius of the Universe and the Hubble expansion (so that, for instance, increasing $a$ indicates  an expansion over the course of cosmological time $t$, etc.). We must emphasize, however, that we will not use any specific property of $h$. In particular, we do not assume that the metric $dt^2-a(t)^2h$ is a solution to Einstein's equations. Therefore, the results of this section are valid for a broader class of metrics than merely those of FLRW models.
\medskip

We can proceed as before until equation (\ref{vcovariante}) which must be modified according to the following calculations: firstly, let us take local coordinates coordinates $x^0=t$ in $\R$ and $x^1,\dots,x^n$ in $M^s$. In this way the coefficients of $g$ are:
$$g_{00}=1,\quad g_{i0}=0,\quad g_{ij}=-a^2h_{ij},\qquad i,j=1,\dots, n,$$
where $h_{ij}$ denotes the $ij$ coefficient of metric $h$. Accordingly,
$$g^{00}=1,\quad g^{i0}=0,\quad g^{ij}=-\frac{1}{a^2}h^{ij},\qquad i,j=1,\dots,n.$$
Now, an easy computation gives us the Cristoffel symbols for connection $\nabla$:
 \begin{alignat*}{3}
\Gamma_{00}^0  & = 0  &\qquad  \Gamma_{0j}^0 &  =0  &\qquad  \Gamma_{ij}^0 &=a\, \dot a\, h_{ij}\\
\Gamma_{00}^k  & = 0  &  \Gamma_{0j}^k  & =\frac{\dot a}a\,\delta_j^k  &  \Gamma_{ij}^k &={\Gamma'}_{ij}^k,
\end{alignat*}
where primed $\Gamma$ denote the symbols corresponding to metric $h$ on $M^s$, and $\dot a$ means $da/dt$.
In this way we get
\begin{alignat*}{2}
&\frac{\partial }{\partial t}^\nabla  \frac{\partial }{\partial t}  = 0
               &\qquad &\frac{\partial }{\partial x^i}^\nabla \frac{\partial }{\partial t} = \frac{\dot a}a\,\frac{\partial }{\partial x^i}
                 \\[10pt]
&\frac{\partial }{\partial t}^\nabla \frac{\partial }{\partial x^j}  = \frac{\dot a}a\, \frac{\partial }{\partial x^j}
             & &\frac{\partial }{\partial x^i}^\nabla \frac{\partial }{\partial x^j} = a\,\dot a\,h_{ij}\frac{\partial }{\partial t}+{\Gamma'}_{ij}^k \frac{\partial }{\partial x^k}
\end{alignat*}
As a consequence,
\begin{align*}
&\frac{\partial }{\partial t}^\nabla v =\frac{\partial }{\partial t}v+\frac{\dot a}a\,v\\[10pt]
&v^\nabla\frac{\partial }{\partial t} =\frac{\dot a}a\,v\\[10pt]
&v^\nabla v =v^{\nabla'}v+a\,\dot a\, h(v,v)\,\frac{\partial}{\partial t}
\end{align*}
where $\nabla'$ is the covariant derivative associated with $h$ in $M^s$  ($v$ is managed as a vector field in $M^s$ depending on the parameter $t$).

Therefore, equation (\ref{vcovariante}) must be replaced by
\begin{equation}\label{vcovariante2}
\overline v^\nabla\overline v=\frac{\partial}{\partial t}v+v^{\nabla'} v+2\,\frac{\dot a}a\,v+a\,\dot a\, h(v,v)\,\frac{\partial}{\partial t}
\end{equation}

Now, we must to insert (\ref{vcovariante2}) into (\ref{Intermediat31}), thus getting
\begin{equation}\label{FLRW1}
\frac{\partial}{\partial t}v+v^{\nabla'} v+2\,\frac{\dot a}a\,v+a\,\dot a\, h(v,v)\frac{\partial}{\partial t}
    +\frac{\textrm{grad}\,P}\rho-\left({\overline v(T(\overline v))+\frac{\overline v(P)}\rho}\right)\,\frac{\partial}{\partial t}=0.
\end{equation}
Taking into account that
\begin{equation*}
\begin{cases}
\displaystyle{\textrm{grad}}P=\frac{\partial P}{\partial t}\,\frac{\partial}{\partial t}-\frac 1a\,\textrm{grad}^hP,\quad\text{($\textrm{grad}^h$ means gradient with respect to $h$)}\\[10pt]
\displaystyle{\overline v(P)=\frac{\partial P}{\partial t}\,\frac{\partial}{\partial t}+v(P)}\\[10pt]
\displaystyle{T\overline v=\frac {(1-a^2h(v,v))}2,\quad\text{so that}\quad}\\
{\hskip 4cm\displaystyle\overline v(T\overline v)=-\frac 12(2a\dot a\,h(v,v)+a^2\,\frac{\partial h(v,v)}{\partial t}+v(a^2h(v,v)))},
\end{cases}
\end{equation*}
Equation (\ref{FLRW1}) becomes
\begin{multline*}
\frac{\partial}{\partial t}v+v^{\nabla'} v+2\,\frac{\dot a}a\,v+a\,\dot a\, h(v,v)\,\frac{\partial}{\partial t}
    -\frac 1{a^2}\,\frac{\textrm{grad}^h\,P}\rho\\-
    \left({-a\,\dot a\,h(v,v)-\frac 12\,a^2\,\frac{\partial h(v,v)}{\partial t}-\frac 12 a^2\,v(h(v,v))
    +\frac{ v(P)}\rho}\right)\,\frac{\partial}{\partial t}=0.
\end{multline*}
Simplifying the previous expression, we obtain
 \begin{multline}\label{FLRW2}
 \frac{\partial}{\partial t}v+v^{\nabla'} v+2\,\frac{\dot a}a\,v+2a\dot a\, h(v,v)\,\frac{\partial}{\partial t}
    -\frac 1{a^2}\,\frac{\textrm{grad}^h\,P}\rho\\+
    \left({\frac 12\,a^2\,\frac{\partial h(v,v)}{\partial t}+\frac 12\,v(h(v,v))
    -\frac{ v(P)}\rho}\right)\,\frac{\partial}{\partial t}=0.
  \end{multline}

 The last step consists of taking the time and spatial components of the above equation: the first one gives a Bernouilli equation and the second one gives an Euler equation:

 \begin{prop}\label{fluidosFLRW}
The intermediate integrals $\overline v$ of a time-constrained mechanical system  (\ref{Newtont4}) 
for  $\alpha=dP/\rho$, with FLRW-type metric $g=dt^2-a(t)^2h$ on $M=\R\times M^s$, such that $\overline v(t)=1$, are the vector fields $\overline v=\partial/\partial t+v$, where $v$ is a tangent field on $M^s$,  satisfying
\begin{equation*}
\frac{\partial}{\partial t}v+v^{\nabla'} v+2\,\frac{\dot a}a\,v-\frac 1{a^2}\,\frac{\textrm{grad}^{\,h}\,P}\rho=0,\quad \text{(Euler equation)}.
\end{equation*}
\end{prop}
\begin{proof}
Equation in the statement is the spatial component of (\ref{FLRW2}). The time component is
$$2\,\frac{\dot a}a\, h(v,v)\,+\frac 12\,\,\frac{\partial h(v,v)}{\partial t}+\frac 12\,v(h(v,v))
  -\frac{ v(P)}{a^2\,\rho}=0$$
  (a Bernoulli equation). Nevertheless, this is a consequence of Euler's equation as seen if we multiply it by $v$ with respect to the metric $h$, similar to what was done in the static case.
\end{proof}

Equation in the above proposition is a general version (for vanishing potential) of the Euler equation for non relativistic fluids in a FLRW type background metric (compare, for instance, with Equation (14) in \cite{Bertschinger} where $\partial/\partial\tau$, $\vec{v}$ and $p$ correspond with our $a\partial/\partial t$, $av$ and $-P$, respectively).

\bigskip

\section{Relativistic fluids as intermediate integrals}\label{s:fluidorelativista}

Now the point of departure is the relativistically constrained mechanical system (\ref{NewtonR})
\begin{equation*}
\iiota_D\omega+dT+\alpha-\frac{\dot\alpha}{\dot\theta}\,\theta=0\ .
\end{equation*}
in the case $\alpha=dP/\zeta$ (the reason for changing $\rho$ for $\zeta$ is because now interpretation is slightly different).

Hence, we have
\begin{equation}\label{IntermediaR}
\iiota_D\omega+dT+\frac{dP}\zeta-\frac{\dot P}{\zeta\,\dot\theta}\,\theta=0.
\end{equation}

In order to apply (\ref{NewtonIntermedia2}) for an intermediate integral $v$ of (\ref{IntermediaR}), we need to compute $v^*\widehat\alpha$ when
$$\widehat\alpha:=\frac{dP}\zeta-\frac{\dot P}{\zeta\,\dot\theta}\,\theta.$$

On the one hand, $v^*dP=dP$, $v^*\zeta=\zeta$ (as forms on $M$); next, $v^*\dot P=v(P)$ by definition of $\dot P$; finally, $v^*\theta=v^\bemol$ and, since $\dot \theta=2T$ (Definition \ref{cinetica}) we get $v^*\dot\theta=2T(v)=g(v,v)$. Putting these values in the above relationship we arrive to:
$$v^*\widehat\alpha=\frac{dP}\zeta-\frac{v(P)}{\zeta\, g(v,v)}\,v^\bemol.$$

In order to obtain convenient expressions for the equations of the intermediate integrals, we will assume that $v$ is \emph{time-like} in the sense of being
$$g(v,v)=\|v\|^2>0;$$
It should be clear that the opposite case, $g(v,v)<0$, can be handled in the same terms, although the final expressions would differ in some signs.

Let us denote $\beta:=\|v\|$ which is a first integral of $v$, because, in relativistic systems (Definition \ref{DefRelativista}), kinetic energy is constant along the solutions. In addition, $v=\beta u$ for a suitable unitary time-like vector field $u$. With this notation, the equation (\ref{NewtonIntermedia2}) (or (\ref{NewtonIntermedia})) for the intermediate integrals of  (\ref{IntermediaR}) is
\begin{equation}
\beta^2 (u^\nabla u)^\bemol+\frac{dP}\zeta-\frac{u(P)}{\zeta}\,u^\bemol=0.
\end{equation}

 Or, in terms of vector fields,
 \begin{equation}
\beta^2 u^\nabla u+\frac{\textrm{grad}\,P}\zeta-\frac{u(P)}{\zeta}\,u=0.
\end{equation}

As a consequence, for $\beta=1$ and putting $\mu:=\zeta-P$ we arrive to the
\begin{prop}
The Euler equations for a relativistic perfect fluid with pressure $P$ and energy density $\mu$,
$$(\mu+P)u^\nabla u+{\textrm{grad}\,P}-{u(P)}\,u=0$$
are the equations of the unitary intermediate integrals $u$  of a mechanical system with force form $\alpha=dP/(\mu+P)$ once it is relativistically corrected.
\end{prop}

 As in previous sections, we use the  adjective relativistic in a broader sense than usual: we admit any type of pseudo-Riemannian metric, in any dimension and with arbitrary signature.
 \bigskip

\section{Conclusions}\label{s:conclusiones}
We have shown that Euler's equations for perfect fluids are the same as those of intermediate integrals of (finite dimensional) mechanical systems, say a material point particle; in each case, the appropriate natural modifications must be made, be they time constraints or relativistic corrections. In this geometric way, a single mathematical object is displayed giving rise to the basic equations of hydrodynamics in its classical (stationary or not) and relativistic versions, for arbitrary pseudo-Riemannian metrics.

It is remarkable that the structure of a finite dimensional mechanical system contains, in itself, that of perfect fluids which are continuum mechanical systems.
More specifically, ensembles of virtual motions of a finite mechanical system perfectly fits with the description of fluid dynamics.
In the case of a relativistic fluid, something else should be highlighted: an approach as a limit case from n-particle systems is problematic, but, as we have seen, the consideration of intermediate integrals (virtual particle beams) of a single particle does leads smoothly and directly to the equations of fluid dynamics.
\bigskip

\begin{acknowledgement*}
The framework in which this article is presented, as well as most of the ideas that have inspired it, are due to my dear friend and thesis advisor, Professor Jesús Muñoz-Díaz, who shared them with me within the scientific seminar he leads in the Mathematics Department of the University of Salamanca. I want to express to him here my deep and sincere gratitude for all that and for so many delightful moments trying to understand a little bit more.
\end{acknowledgement*}

\bigskip

\end{document}